\documentclass[a4paper, 10 pt]{article}

\AtBeginDocument{%
 }

\usepackage{tikz-cd}
\usepackage{tikz}
\usetikzlibrary{automata, arrows.meta, positioning}
\usepackage{float}
\usepackage{bussproofs}
\usepackage{graphicx}%
\usepackage{multirow}%
\usepackage{amsmath}%
\usepackage{amsthm}%
\usepackage{mathrsfs}%
\usepackage{amssymb}%
\usepackage[title]{appendix}%
\usepackage{xcolor}%
\usepackage{textcomp}%
\usepackage{manyfoot}%
\usepackage{algorithm}%
\usepackage{algorithmicx}%
\usepackage{algpseudocode}%
\usepackage{enumerate}
\usepackage{physics}
\usepackage{hyperref} 
\usepackage{wrapfig}
\usepackage[margin=1 in, a4paper]{geometry}
\hypersetup{
 colorlinks=true,
 linkcolor=blue,
 filecolor=magenta,      
 urlcolor=cyan,
 pdftitle={Overleaf Example},
 pdfpagemode=FullScreen,
 }
\theoremstyle{definition}
\newtheorem{dfn}{Definition}[]

\newtheorem{prop}[dfn]{Proposition}
\newtheorem{rem}[dfn]{Remark}
\newtheorem{lem}[dfn]{Lemma}
\newtheorem{cor}[dfn]{Corollary}

\newtheorem{them}[]{Theorem}

\newcommand{\bol}{\mathbf}

\newenvironment{bprooftree}
 {\leavevmode\hbox\bgroup}
 {\DisplayProof\egroup}

\title{Non-commutative linear logic fragments\\ with sub-context-free complexity
}

\author{Yusaku Nishimiya$^{1,2}$ \and Masaya Taniguchi$^2$}
\date{%
    \begin{small}
    $^1$University of Illinois Springfield, IL, USA\\%
    $^2$RIKEN Center for Advanced Intelligence Project (AIP),\\ Tokyo, Japan\\[2ex]%
    \end{small}
}

\begin{document}
\maketitle

\begin{abstract}
We present new descriptive complexity characterisations of classes \textsf{REG} (regular languages), \textsf{LCFL} (linear context-free languages) and \textsf{CFL} (context-free languages) as restrictions on inference rules, size of formulae and permitted connectives in the Lambek calculus; fragments of the intuitionistic non-commutative linear logic with direction-sensitive implication connectives.
Our identification of the Lambek calculus fragments with proof complexity \textsf{REG} and \textsf{LCFL} is the first result of its kind. We further show the \textsf{CFL} complexity of one of the strictly `weakest' possible variants of the logic, admitting only a single inference rule. The proof thereof, moreover, is based on a direct translation between type-logical and formal grammar and structural induction on provable sequents; a simpler and more intuitive method than those employed in prior works. We thereby establish a clear conceptual utility of the Cut-elimination theorem for comparing formal grammar and sequent calculus, and identify the exact analogue of the Greibach Normal Form in Lambek grammar.
We believe the result presented herein constitutes a first step toward a more extensive and richer characterisation of the interaction between computation and logic, as well as a finer-grained complexity separation of various sequent calculi.
\end{abstract}


\subsection*{Introduction}
\begin{figure}[h!]
  \vspace*{- 1 em}
  \begin{small}
    \begin{center}
      \begin{bprooftree}
        \AxiomC{$\alpha,\alpha \to \beta$}
        \RightLabel{Contraction}
        \UnaryInfC{$\alpha \to \beta$}
      \end{bprooftree}
      \quad
      \begin{bprooftree}
        \AxiomC{$\alpha \to \gamma$}
        \RightLabel{Weakening}
        \UnaryInfC{$\alpha,\beta \to \gamma$}
      \end{bprooftree}
      \quad restricted in \textbf{LL}.
    \end{center}
    
    \begin{center} 
      \begin{bprooftree}
        \AxiomC{$\Gamma, \alpha, \beta, \Delta \to \gamma$}
        \RightLabel{Exchange}
        \UnaryInfC{$\Gamma, \beta, \alpha, \Delta \to \gamma$}
      \end{bprooftree}
      \quad
      \begin{bprooftree}
        \AxiomC{$\Gamma, \alpha,\Theta \to \beta; $}
        \AxiomC{$\Delta \to \alpha$}
        \RightLabel{Cut}
        \BinaryInfC{$\Gamma, \Delta, \Theta \to \beta$}
      \end{bprooftree}
      \quad allowed in \textbf{LL}.
    \end{center}
  \end{small}
  \vspace*{- 1.7 em}
\end{figure}
Proof theorists study substructural logics to understand the effect of admitting or eliminating \textit{structural rules} on the properties of proof systems, usually presented as a sequent calculus.
Of particular interest for computation is linear logic (\textbf{LL}) \cite{girard1987linear} as it restricts Contraction and Weakening rules, making proofs more \textit{resource-conscious} (and thus computation-relevant) than its classical counterpart.

The multiplicative-additive fragment of linear logic (\textbf{MALL})\footnote{Here, we focus exclusively on propositional logic without the \textit{exponential} connectives. For a survey of computability and complexity for more general linear logic and its first-order extension, see~\cite{lincoln1995deciding}.}~\cite{lincoln1992decision}  in which each formula is used exactly once was shown to be \textsf{PSPACE-complete}~\cite{lincoln1992decision} and further restriction by removal of additives, to multiplicative linear logic (\textbf{MLL}) makes the calculus \textsf{NP-complete} \cite{kanovich1991multiplicative}.

Little is known about the computational complexity of proof systems based on even `weaker' fragments of \textbf{LL}, except for the \textsf{NP-complete}ness \cite{pentus2006lambek} of the Lambek calculus (\textbf{L}), the intuitionistic, non-commutative, multiplicative fragment of \textbf{LL} with direction-sensitive implications. \textbf{L} was originally introduced as a proof system for formalising natural language syntax in Lambek's seminal paper~\cite{lambek1958mathematics} but was later shown by Abrusci \cite{abrusci1990comparison} to be a fragment of \textbf{LL} without Exchange, leaving Cut as the sole structural rule.
Chomsky conjectured in~\cite{chomsky1963formal} the equivalence between type-logical grammars based on \textbf{L} and context-free grammars (CFG) and thus began the research into the expressivity of Lambek grammar. 
Pentus confirmed Chomsky's conjecture~\cite{pentus1993lambek} and further proved~\cite{pentus1997product} that removal of multiplicative connective does not change the expressivity, such that the \textit{product-free} Lambek grammar is also context-free. However, no fragments of \textbf{L} corresponding in expressivity to lower classes of formal grammars (equivalently automata) in the Chomsky hierarchy have been identified.

Here, we show that, intuitively, the restriction on the size and directionality of the logical formulae permitted in the proof system, more so than on inference rules, yields the variation in expressivity that corresponds exactly to the difference between context-free, linear context-free and regular grammar. 

\subsection*{Preliminaries}
\subsubsection*{Formal languages}
The expressive power of a class of automata is defined by the formal languages it can parse (i.e. decide the set membership). Formal languages, in turn, may be characterised by the nature of grammatical \textit{production rules} required to generate all strings therein. Here, we consider the class of context-free grammars.

\begin{dfn}
 A context-free grammar (CFG) $\bol{G} = (N, \Sigma, S_\bol{G}, P)$ consists of sets $N$ of \textit{non-terminal symbols}, $\Sigma$ of \textit{terminal symbols}, a \textit{start symbol} $S_\bol{G} \in N$ and $P \subset N \times (\Sigma \cup N)^+$, the set of production rules, where $(\Sigma \cup N)^+$ consists of finite non-empty strings of terminal and non-terminal symbols.

 Any rule $p \in P$ of the form $A \to aB$ (resp. $A \to Ba$) is said to be \textit{right-linear} (resp. \textit{left-linear}).
  
 A CFG is a linear context-free grammar (LCFG) if all production rules are either right or left-linear.

 A (right-) regular grammar (REG) is a context-free grammar all of whose production rules are right-linear. Likewise and equivalently for the left-regular grammar.
\end{dfn}

Intuitively, the context-freeness signifies the independence of the rewriting-rule applicability from the surrounding symbols, whilst the linearity means that the length of the intermediate terminal/non-terminal string increases by at most one during the generation.

\subsubsection*{Lambek calculus}
We define the Lambek calculus \textbf{L} as a calculus with an axiom and rules acting on sequents of the form $\Gamma \to \alpha$.
We shall interchangeably use the words \textit{type} and \textit{formula} to denote propositions.
Capital Greek letters represent a finite sequence of types, lowercase Greek letters a (derived) type and capital Latin letters a primitive type in $Pr$, the finite set of primitive types.\footnote{We use the word `string' for concatenated products of symbols in the alphabet and non-terminal symbols, and `sequence' for comma-separated lists of types.}
For any sequence of types $\Gamma$, we let $|\Gamma|$ be the number of types in the sequence $\Gamma$.
In all systems considered here, all type-connectives are binary, namely, $/,\backslash$ and $\cdot$.
The set of all types $Tp$ is thus defined as the smallest set such that i. $Pr \subseteq Tp$ and ii. for any types $\alpha, \beta \in Tp$, therefrom-derived types, $\alpha / \beta, \alpha \backslash \beta, \alpha \cdot \beta$ are in $Tp$.\footnote{Type constructions are non-associative such that, in fully parenthesised notation, derived types are $(\alpha) / (\beta), (\alpha) \backslash (\beta), (\alpha) \cdot (\beta)$.}
We let $Tp(/)$ be the set of all types in which $/$ is the only type connective that occurs, and likewise for $Tp(\backslash)$ and $Tp(\cdot)$. 
The \textit{degree} of a type $\alpha$, denoted $d(\alpha)$, is the number of distinct occurrences of connectives in $\alpha$, intuitively, a measure for the size of types. 

For $n \in \mathbb{N}$, we let $Tp_n$ be the set of types whose degree is less than or equal to $n$.
\textbf{L} consists of one Axiom, the Cut-rule and six inference rules shown below.
\begin{figure}[h!]
  \vspace*{0 em}
  \begin{center} 
    The Lambek calculus.\\\vspace*{0.5 em}

  \begin{bprooftree}
      \AxiomC{$\vphantom{\Gamma,;}$}
      \RightLabel{Axiom}
      \UnaryInfC{$\alpha \to \alpha$}
  \end{bprooftree}
  \qquad
  \begin{bprooftree}
    \AxiomC{$\Gamma, \alpha,\Theta \to \beta; $}
    \AxiomC{$\Delta \to \alpha$}
    \RightLabel{Cut}
    \BinaryInfC{$\Gamma, \Delta, \Theta \to \beta$}
  \end{bprooftree}
  \end{center}

  \begin{center}
  \begin{bprooftree}
    \AxiomC{$\alpha, \Gamma \to \beta$}
    \RightLabel{($\to \backslash$), where $\Gamma \neq \Lambda$}
    \UnaryInfC{$\Gamma \to \alpha \backslash \beta$}
  \end{bprooftree}
  \qquad
  \begin{bprooftree}
    \AxiomC{$\Gamma \to \alpha;$}
    \AxiomC{$\Delta, \beta, \Theta \to \gamma$}
    \RightLabel{$(\backslash \to)$}
    \BinaryInfC{$\Delta, \Gamma, (\alpha \backslash \beta), \Theta \to \gamma$}
  \end{bprooftree}
  \end{center}

  \begin{center}
  \begin{bprooftree}
    \AxiomC{$\Gamma, \alpha \to \beta$}
    \RightLabel{($\to /$), where $\Gamma \neq \Lambda$}
    \UnaryInfC{$\Gamma \to \beta/\alpha$}
  \end{bprooftree}
  \qquad
  \begin{bprooftree}
    \AxiomC{$\Gamma \to \alpha;$}
    \AxiomC{$\Delta, \beta, \Theta \to \gamma$}
    \RightLabel{$(/ \to)$}
    \BinaryInfC{$\Delta, (\beta/\alpha), \Gamma, \Theta \to \gamma$}
  \end{bprooftree}
  \end{center}

  \begin{center}
  \begin{bprooftree}
    \AxiomC{$\Gamma \to \alpha;$}
    \AxiomC{$\Delta \to \beta$}
    \RightLabel{($\to \cdot$)}
    \BinaryInfC{$\Gamma, \Delta \to \alpha \cdot \beta$}
  \end{bprooftree}
  \qquad
  \begin{bprooftree}
    \AxiomC{$\Gamma, \alpha, \beta, \Delta \to \gamma$}
    \RightLabel{$(\cdot \to)$}
    \UnaryInfC{$\Gamma, (\alpha \cdot \beta),\Delta \to \gamma$}
  \end{bprooftree}\\
  \vspace{0.5 em}
  \begin{small}
     We let $\Lambda$ be the empty sequence of types. 
  \end{small}
  \end{center}
  \vspace*{- 1 em}
\end{figure}

A \textit{fragment} of \textbf{L} is a sequent calculus with the Axiom, Cut and some but not all of the six inference rules. We let $\bol{L}(/ \to, \to /)$ denote the fragment of \textbf{L} with $(/ \to)$ and $(\to /)$ rules, as an example. We now define Lambek grammar.
\begin{dfn}
  A Lambek grammar $\mathcal{G}$ is a quadruple $(Pr,V, S_\mathcal{G}, f)$, with the set of primitive types $Pr$, the finite set of symbols or \textit{alphabet} $V$, the \textit{distinguished type} $S_\mathcal{G} \in Pr$ and the type assignment function $f:V \to \Omega^{Tp}$, where $\Omega^{Tp}$ is the powerset of $Tp$. $f$ is naturally extended to strings; $f^+:V^+ \to \Omega^{Tp^+}$ 
  defined by $\forall w \in V^+ \text{ s.t. } w = a_1 ...a_n, \quad  f^+(w) = \{\Gamma \in Tp^+ | \Gamma = \alpha_1 ... \alpha_n \text{ s.t. } \forall k, \alpha_k \in f(a_k) \} $ 
  where $V^+$ is the set of all finite non-empty strings of symbols in $V$ and $Tp^+$ is the set of all finite non-empty sequences of types.

  The \textit{language $\mathcal{L}$ recognised by} $\mathcal{G}$ is a subset of $V^+$, such that for any $w \in V^+$, $w \in \mathcal{L}$ iff $\exists \Gamma \in f^+(w)$ and $ \bol{L} \vdash \Gamma \to S_\mathcal{G}$ (\textit{i.e.} any given string is in the language iff there is a sequence of types assigned to it which is \textit{reducible} to $S_\mathcal{G}$ in $\bol{L}$).
  The grammar and language for fragments of $\bol{L}$ are analogously defined.
\end{dfn}

In the present work, we identify the fragments of Lambek calculus $\bol{L}$ with equivalent expressivity to context-free grammar subclasses by constructing a suitable Lambek grammar. 

\subsection*{Main results}
The construction of corresponding grammars relies on structural inductions on the sequent, which in turn requires the existence of a Cut-free proof for any provable sequents, guaranteed by the \textit{Gentzen's Theorem} in \cite{lambek1958mathematics}.
\begin{them}\label{Cut-Elim} \cite{lambek1958mathematics}
 The elimination of Cut from $\bol{L}(/\to)$ does not change the set of provable formulae and likewise holds for $\bol{L}(/ \to, \backslash \to)$.
\end{them}
\vspace*{- 0.5 em}
\noindent\textit{Sketch of proof.} To illustrate, we present the sequent replacement procedure for  $\bol{L}(/\to)$. 
\begin{center}
  \begin{bprooftree}
    \AxiomC{$\Gamma, \alpha,\Theta \to \beta; $}
    \AxiomC{$\Delta \to \alpha$}
    \RightLabel{Cut}
    \BinaryInfC{$\Gamma, \Delta, \Theta \to \beta$}
  \end{bprooftree}
  $\quad \Rightarrow \quad$
  \begin{bprooftree}
  \AxiomC{$\Gamma, \alpha, \Theta \to \beta;$}
  \AxiomC{$\Delta' \to \alpha$}
  \RightLabel{Cut'}
  \BinaryInfC{$\Gamma, \Delta', \Theta \to \beta ; $}
  \AxiomC{$\Xi \to \alpha'$}
  \RightLabel{$(/\to)$}
  \BinaryInfC{$\Gamma, \Delta, \Theta \to \beta $}
\end{bprooftree}
\end{center}
Assume that the premises of the Cut on the left are provable without Cut. Then, the last step in the derivation of $\Delta \to \alpha$ is the $(/\to)$ as shown below. 
\begin{center}
  \begin{bprooftree}
      \AxiomC{$\Delta' \to \alpha;$}
      \AxiomC{$\Xi \to \alpha'$}
      \RightLabel{$(/\to)$}
      \BinaryInfC{$\Delta \to \alpha$}
  \end{bprooftree}
\end{center}
Readers can verify that the replacement of Cut by a `smaller' Cut (Cut' on the right) is possible due to the assumed Cut-free provability of relevant premises, noting in particular that $\Delta'$ contains one less $/$ connective than $\Delta$. One can thus repeat this until all premises are instances of Axiom. The procedure is analogous for $\bol{L}(/ \to, \backslash \to)$. We now state our main theorem.

\begin{them}\vspace*{- 0 em}
 The following three pairs of Lambek-fragment grammar and formal grammar (without the empty string, $\epsilon$) are of equivalent expressive power.
 \begin{align*}
     \text{$\bol{L}(/\to)$-grammar with $Tp(/)$ } &\Leftrightarrow \text{CFG} \\
     \text{$\bol{L}(/ \to, \backslash \to)$-grammar with $Tp_1(/,\backslash)$} &\Leftrightarrow \text{LCFG} \\
     \text{$\bol{L}(/ \to)$-grammar with $Tp_1(/)$} &\Leftrightarrow \text{REG}
 \end{align*}
\end{them}

Our proof consists of the two following steps.\\
I. Identification of appropriate type assignments (resp. production rules) given an $\epsilon$-free CFG (resp. Lambek grammar).
\\
II. Showing their correspondence via structural induction on provable sequents whose consequent is the distinguished type.

And it follows an observation of what constraints on provable sequents result from 
the restriction of the calculus to $\bol{L}(/\to)$, which illustrates the conditions under which the antecedent $\Gamma$ {reduces} to $S_\mathcal{G}$.

\begin{lem}\label{lem:uni-reducible}
 (Reducibility condition) 
 Let $\Gamma$ be a non-empty sequence of types in $Tp(/)$.\\
  $\bol{L}(/\to) \vdash \Gamma \to S_\mathcal{G}$ iff
  $\Gamma = \alpha, \Delta_1, ...,\Delta_n$ where
  \par 1. $\alpha$ is of the form $ ( \cdots ((S/\beta_n)/\beta_{n-1})/ \cdots)/\beta_1$ where $\beta_1,..,\beta_n \in Tp(/)$ and 
  \par 2. for all $1 \leq k \leq n$, $\bol{L}(/\to) \vdash \Delta_k \to \beta_k$.

 Moreover, likewise holds for reducibility to any other types besides $S_\mathcal{G}$.
\end{lem}

\begin{proof}
  Requirement for the form of the leading type $\alpha$ with the left-most type being $S_\mathcal{G}$ is clear from the manner in which the $(/\to)$ rule acts on a type, namely, appending a type on the right after the $/$ symbol. The remainder of the lemma stipulates that the $\alpha$ is followed by a sequence that `reflects' the structure thereof. Formally, we can observe in the inference of the form
  \begin{center}
    \begin{bprooftree}
      \AxiomC{$\Delta_1 \to \alpha_1 ;$}
      \AxiomC{$( \cdots ((S_\bol{G}/\alpha_n)/\alpha_{n-1})/ \cdots)/\alpha_2, \Delta_2, ...,\Delta_n \to S_\mathcal{G}$}
      \RightLabel{$(/\to)$}
      \BinaryInfC{$( \cdots ((S_\bol{G}/\alpha_n)/\alpha_{n-1})/ \cdots)/\alpha_1, \Delta_1, ...,\Delta_n \to S_\mathcal{G}$}
    \end{bprooftree}
  \end{center}
  that for the degree of the leading type to be reduced (seeing from bottom to top) by $(/ \to)$, it is sufficient and necessary that the leading type be followed by a sequence $\Delta_1$ such that $\Delta_1 \to \alpha_1$ is provable. The lemma follows from induction on the $n$-steps to reduce $\alpha$ to $S_\mathcal{G}$.
\end{proof}

We now show the language equivalence separately for each class.

\begin{prop}\label{CFG-Lambek}\vspace*{- 0.5 em}
  Lambek grammars based on $\bol{L}(/\to)$ with $Tp(/)$ ($\bol{L}(/\to)$-grammars) recognise exactly context-free languages without the empty string.
\end{prop}

\begin{proof} \phantom{ }\\
(CFG $\Rightarrow$ $\bol{L}(/\to)$ grammar)

Let $\bol{G} = (N, \Sigma, P, S_\bol{G})$ be a CFG recognising an $\epsilon-$free language.
Assume Greibach Normal Form~\cite{greibach1965new}, such that all production rules are $A \to a $ or $A \to aB_1 \cdots B_n$ for some $A, B_1, ..., B_n \in N$ and $a \in \Sigma$.
Consider an $\bol{L}(/\to)$ grammar $\mathcal{G} = (Pr, V, S_\mathcal{G}, f)$ constructed by identifying
$Pr = N$,
$V = \Sigma$,
$S_\mathcal{G} = S_\bol{G}$ (hereafter $S$) and defining $f:V \to \Omega^{Tp(/)}$ as follows.
For any $a \in V$, $f(a) \subseteq Tp(/)$ is the smallest set such that
$A \in f(a)$ if $A \to a \in P$ and
$(\cdots((A/B_n)/B_{n-1})/\cdots)/B_1 \in f(a)$
if $A \to aB_1 ... B_{n-1}B_n \in P $.

Lemma \ref{lem:uni-reducible} implies that any type assigned to some symbol which is potentially reducible to $S$ corresponds to production rules with $S$ on the left. Thus, consider such a production rule $S \to a_1 A_1\cdots A_n$.
Note that such a rule stipulates that if the derivation of some string $w = a_1...a_n$ begins with the rule, then some string which can be generated from $A_1$ shall be on the right of $a_1$. For each step in the generation of some string $w'$ from $A_1$, there shall be a corresponding type-assignment. Recursive application of Lemma \ref{lem:uni-reducible} to all such rules and induction on the finite length of generation implies language equivalence. Let us now consider the converse.\\
($\bol{L}(/\to)$ grammar $\Rightarrow$ CFG)

Let $\mathcal{G} = (Pr, V, S_\mathcal{G}, f)$ be an arbitrary $\bol{L}(/\to)$-grammar. Construct a CFG, $\bol{G} = (N, \Sigma, P, S_\bol{G})$  by letting: $\Sigma = V$, $S_\bol{G} =  S_\mathcal{G} $, $N = \overline{f}(V)$, the set of all sub-types\footnote{Considering a type $\alpha$ as a string, its substring $\alpha'$ is a subtype of $\alpha$ if it is a type.} of all types assigned to symbols in $V$ by $f:V \to \Omega^{Tp(/)}$ which is defined as follows. 
  For any $a \in V$
    \begin{itemize}
      \item[] $\alpha \to a\beta_1 \beta_2 \cdots \beta_n \in P$ if $(\cdots ((\alpha/\beta_n)/\beta_{n-1})/\cdots)/\beta_1 \in f(a)$ and
      \item[] $\alpha \to a \in P$ if $\alpha \in f(a)$ .
    \end{itemize}
Let us see that the grammars thereby constructed recognise the same language.
Consider a production rule of the form
$S_\bol{G} \to a\beta_1 \beta_2 \cdots  \beta_n$ with $S_\bol{G}$ on the left-hand side.
The application of Lemma \ref{lem:uni-reducible} to the corresponding type assignment
$( \cdots ((S_\bol{G}/\beta_n)/\beta_{n-1})/ \cdots) /\beta_1 \in  f(a)$ and recursively to the type assignments that correspond to production rules with $\beta_1, \beta_2,...$ or $\beta_n$ on the left-hand side and so on, implies the language equivalence by induction on the length of Cut-free derivation of any provable sequents assignable to strings.
\end{proof}

\begin{rem}
  The key difference between the proof of CFG $\Rightarrow \bol{L}(/\to) $ and  $\bol{L}(/\to)\Rightarrow $ CFG  is the definition of non-terminals $N$. 
  This is due to the fact that in the latter, we must take into account {derived types of all possible forms}, whereas in the former, we may construct those types derivable by introducing only $/$ followed by a \textit{primitive} type, which means that the change of degree is exactly by one with every application of $(/\to)$.
\end{rem}

The construction of a Lambek grammar for linear context-free languages is more straightforward as the relevant types are smaller.

\begin{prop}\label{linear-prop}
  $\bol{L}(/\to, \backslash \to)$-grammars with types restricted to $Tp_1(/,\backslash)$ recognise exactly linear languages.
\end{prop}
\begin{proof}
  (Construction)
  Let $\bol{G} = (N, \Sigma, P, S_\bol{G}) $ be an $\epsilon-$free LCFG.
  We construct a corresponding $\bol{L}(/\to, \backslash \to)$-grammar $\mathcal{G} = (Pr, V, S_\mathcal{G}, f) $ with types of degree less than or equal to 1, and vice versa, by letting: $Pr = N$, $V = \Sigma$, $S_\mathcal{G} = S_\bol{G}$ (hereafter, simply $S$) and $f:V \to Tp_1(/,\backslash)$ (conversely, $P$) be defined as follows. For any $a \in V$
  \begin{itemize}
    \item[] $A/B \in f(a)$ iff $A\to aB \in P$
    \item[] $B \backslash A \in f(a)$ iff $A \to Ba \in P$ and
    \item[] $A \in f(a)$ iff $A \to a \in P$.
  \end{itemize}
  (Correctness) Let $w = a_1...a_n$. By definition, $w \in \mathcal{L(G)}$ iff $\exists \Gamma \in f^+(w) \subset Tp_1(/,\backslash)^n$ such that $\Gamma \to S$ is provable using only $(/ \to)$, $(\backslash \to)$ and Cut.
  Note that $\bol{L}(/ \to, \backslash \to) \vdash \Gamma \to S$ iff $\Gamma$ has the form;
  $S/A, \Delta$,  
  $\Delta, A\backslash S$
  for some $A \in Pr$ and $\Delta \in Tp_1(/,\backslash)^+$ such that $\bol{L}(/ \to, \backslash \to) \vdash \Delta \to A$ (due to the eliminabilty of Cut), if not trivially $\Gamma = S$.
  Thus, when $\Gamma \neq S$, $\bol{L}(/ \to, \backslash \to) \vdash \Gamma \to S$ iff $P$ includes a production rule of the form either
  $S \to a_1 A \in P$ or $S \to A a_n \in P$. 
 
  \begin{center}
    \begin{bprooftree}
        \AxiomC{$\vdots$}
        \UnaryInfC{$\Delta \to A;$}
        \AxiomC{$S \to S$}
        \RightLabel{$(/\to)$}
        \BinaryInfC{$S/A, \Delta \to S$}
    \end{bprooftree}
    \quad or \quad 
    \begin{bprooftree}
        \AxiomC{$\vdots$}
        \UnaryInfC{$\Delta \to A;$}
        \AxiomC{$S \to S$}
        \RightLabel{$(\backslash \to)$}
        \BinaryInfC{$\Delta, A\backslash S \to S$}
    \end{bprooftree}
 \end{center}
Likewise holds for the provability of the sequent $\Delta \to A$.
Now, as $\Gamma$ is finite and $|\Delta| < |\Gamma|$, we shall find a finite number of corresponding applications of production rules with each proof step by $(/\to)$ or $(\backslash \to)$ until the length of the antecedent is reduced to one.
\end{proof}

\begin{cor}
  $\bol{L}(/\to)$-grammars with types restricted to $Tp_1(/)$ recognise exactly regular languages.
\end{cor}
\begin{proof}
  The construction of corresponding grammars is identical to that in the proof of Proposition \ref{linear-prop} except for the omission of left-linear rules. Correctness likewise follows from the analogous argument.
\end{proof}

\subsection*{Discussion}
The result presented here shows the Lambek grammar's sensitivity to the restrictions on the type degree and directionality.
The unidirectional $\bol{L}(/\to)$-grammar is the simplest Lambek grammar with context-free complexity and may be considered as the proof-theoretic analogue of the Greibach normal form~\cite{greibach1965new}. We further note that the type degree restriction to one naturally corresponds to the (bi)linearity of production rules. Our primary contribution here is thus the direct translation between inference rules of logic and production rules of formal grammar.
Though the language equivalences themselves are not particularly surprising, we believe the directness of the correspondence
equips us with an intuition to extend the result to related and more general classes of interesting problems. 
Promising directions include i. further work on fine-grained descriptive complexity of other (non-commutative) linear logic fragments, ii. identification of Lambek grammars equivalent to star-free languages, mildly context-sensitive languages, Lindenmayer systems~\cite{lindenmayer1972developmental}, etc. and iii. the interaction between the semantics of the linear logic and geometric group theoretic characterisations of formal languages facilitated by an analogous construction of a type-logical grammar.

\newpage 
\bibliographystyle{alpha}
\bibliography{ref.bib}

\subsection*{Acknowledgements}
This work was supported by JSPS KAKENHI Grant Number 24K16077.
YN thanks the Neural Circuits and Computations Unit, RIKEN Center for Brain Science, for providing a friendly working space and colleagues from various centres of RIKEN for stimulating questions. The authors further acknowledge Naoki Negishi for participating in our weekly discussions and thank the anonymous reviewers for fruitful recommendations that led to the present form of this paper.
\end{document}